\documentclass[12pt]{article}
\usepackage{graphicx} 
\usepackage{amssymb}
\usepackage{amsmath, amsthm}
\usepackage{appendix}
\usepackage{parskip}
\usepackage{float} 
\usepackage{caption}
\usepackage{mathtools}  
\mathtoolsset{showonlyrefs}
\usepackage[colorlinks=true, linkcolor=blue, citecolor=blue, urlcolor=blue]{hyperref}
\usepackage[margin=2.67cm]{geometry}
\newcommand{\norm}[1]{\left \lVert #1 \right \rVert}

\title{Assessing Quantum Advantage for Gaussian Process Regression}
\author{
    Dominic Lowe$^{1}$, 
    M.S. Kim$^{1}$, 
    Roberto Bondesan$^{2}$ 
}
\date{}
\newtheorem{lemma}{Lemma}
\newtheorem{theorem}{Theorem}
\newtheorem{corollary}{Corollary}[theorem]
\newtheorem*{namedtheorem}{Theorem}
\newtheorem*{dataassumption}{Assumption}
\begin{document}

\maketitle
\begin{center}
$^{1}$ Blackett Laboratory, Imperial College London, SW7 2AZ, United Kingdom\\
$^{2}$ Department of Computing, Imperial College London, SW7 2AZ, United Kingdom
\end{center}

\begin{abstract}
    Gaussian Process Regression is a well-known machine learning technique for which several quantum algorithms have been proposed. 
    We show here that in a wide range of scenarios these algorithms show no exponential speedup.
    We achieve this by rigorously proving that the condition number of a kernel matrix scales at least linearly with the matrix size under general assumptions on the data and kernel.
    We additionally prove that the sparsity and Frobenius norm of a kernel matrix scale linearly under similar assumptions. 
    The implications for the quantum algorithms runtime are independent of the complexity of loading classical data on a quantum computer and also apply to dequantised algorithms.
    We supplement our theoretical analysis with numerical verification for popular kernels in machine learning.
\end{abstract}

\section{Introduction}

Over the last decade, the question of whether quantum computers can speed up machine learning workloads has been under intense investigation \cite{Biamonte_2017,Cerezo_2022,Ciliberto_2018}.
With quantum computers excelling at performing matrix operations in high-dimensional spaces, a natural question is whether machine learning algorithms can be implemented using quantum linear algebra more efficiently
\cite{lloydQuantumAlgorithmsSupervised2013,lloydQuantumPrincipalComponent2014,kerenidisQuantumRecommendationSystems2017}.
The workhorse of these proposals is the HHL algorithm that approximates the solution of a linear system \cite{harrowQuantumAlgorithmSolving2009}.
An important caveat of these quantum machine learning algorithms is that for exponential speedup over the best classical methods, we need to load the classical data into a quantum superposition in logarithmic time, an assumption whose practicality has been heavily debated \cite{jaques2023qramsurveycritique}.
In fact, classical algorithms supplied with a similar input data structure exhibit only a polynomial slowdown compared to their quantum analogues \cite{tangQuantuminspiredClassicalAlgorithm2019,tangQuantumPrincipalComponent2021}. 

In this paper, we consider the problem of benchmarking quantum machine learning algorithms assuming the data has been loaded on a quantum computer. We focus primarily on the problem of Gaussian Process Regression (GPR).
GPR is a machine learning technique that is widely applied due to its ability to provide calibrated uncertainty over predictions \cite{rasmussenGaussianProcessesMachine2005}, but that suffers from a computational bottleneck due to the requirement of inverting a covariance matrix.
We show that for a broad class of kernels, the condition number of the covariance matrix -- defined as the ratio of its largest to smallest eigenvalue -- scales at least linearly with the number of data points. All currently proposed quantum algorithms for Gaussian Process Regression (GPR) have a time complexity that depends, either linearly or polynomially, on this condition number
\cite{zhaoQuantumassistedGaussianProcess2019,chenQuantumAlgorithmGaussian2022,farooqQuantumassistedHilbertspaceGaussian2024a}.

The significance of this result is that it allows a direct comparison of these quantum algorithms to their classical counterparts. From this it becomes clear that these algorithms (even before data-loading or error correction are taken into account) can only perform marginally better than a standard classical approach.  

Several alternative classical approaches to matrix inversion also have a time complexity that depends at least linearly on the condition number, and so are affected similarly. For example, conjugate gradient methods for sparse matrices \cite{ConjugateGradientDescent}
and quantum inspired approaches for matrix inversion \cite{shaoFasterQuantuminspiredAlgorithms2023}. The algorithms mentioned in \cite{shaoFasterQuantuminspiredAlgorithms2023} instead depend on the condition number with respect to the Frobenius norm, but this still scales at least linearly with the number of datapoints by corollary \ref{cor:2}.

Our main technical contribution is an asymptotic formula for the condition number in the regularised case, and a lower bound in the un-regularised case. We believe the former is the first result of its kind, in that it describes the condition number (for a wide variety of kernels) purely in terms of the regularisation parameter and the largest eigenvalue of the corresponding integral operator.
Our results show that, independently of the practicality of loading the classical data on a quantum computer in logarithmic time, achieving exponential quantum advantage for GPR requires novel approaches.

\subsection{Overview of Gaussian Process Regression}
GPR is a non-parametric machine learning technique with diverse applications \cite{rasmussenGaussianProcessesMachine2005}. Given training data in the form of input-output pairs, it models the underlying function as a Gaussian Process. The task is then to determine the predictive distribution at a test point given training data. From here samples of the function evaluated at the test points can be drawn, yielding predicted outputs. 

Suppose we have some training data $\{x_i, y_i\}_{i=1}^{m}$ with inputs $x_i \in \chi$, where $\chi$ is a measurable space and noisy outputs $y_i \in \mathbb{R}$. We aim to model a function $f(x)$ that satisfies:
$$y = f(x) + \epsilon$$
Here $\epsilon$ denotes the noise in the outputs and we assume that $\epsilon \sim N(0,\sigma_n^2)$.

A Gaussian process is a collection of random variables, any finite number of which have a joint Gaussian distribution. In the scenario described here we treat $f(x_i)$ as Gaussian random variables. The Gaussian process has a prior distribution specified by a mean function $m(x)$ and a covariance function $k(x, x')$ given by:
\begin{align}
    m(x) &= \mathbb{E}[f(x)] \\
    k(x,x') &= \mathbb{E}[((f(x)- m(x))(f(x') - m(x'))]
\end{align}
Both functions can be chosen by the user. Typically a choice of $m(x) \equiv 0$ is used, but if there is prior information suggesting otherwise then this can be altered. The covariance function (or kernel) specifies how closely related two function values should be at different distances. Generally, the covariance function determines how smooth and stationary the considered functions should be.  

For a single test point $x_*$, the distribution of $f(x_*)$ (denoted by $f_*$ for simplicity), conditional on the training data and $x_*$, is Gaussian. We denote the mean and variance of this distribution by $\overline{f_*}$ and $\mathbb{V}[f_*]$ respectively, such that:
$$ p(f_* | x_*, \{x_i, y_i\}) \sim N(\overline{f_*}, \mathbb{V}[f_*])$$
We let $K$ denote the $m\times m$ matrix $[k(x_i,x_j)]_{i,j}$ where $i,j$ range over all the training data. We also let $\mathbf{k_*}$ denote the vector $[k(x_i,x_*)]_{i=1}^m$. Then with the choice of mean function $m(x) = 0$ it can be shown that:
\begin{align}
    \overline{f}_* &= \mathbf{k}_*^T(K+\sigma_n^2I)^{-1}\bf{y} \\
    \mathbb{V}[f_*] &= k(x_*,x_*) - \mathbf{k}_*^T(K+\sigma_n^2I)^{-1}\mathbf{k}_*
\end{align}
where $\bf{y}$ is the vector of training outputs. A full derivation can be found in \cite{rasmussenGaussianProcessesMachine2005}. Calculating these is how inference is performed using GPR.

\subsection{Quantum Algorithms for GPR}
It is clear that the largest computational cost of performing GPR is the matrix inversion of $K +\sigma_n^2 I$. In a classical setting this is typically performed via a Cholesky decomposition, leading to a time complexity of $O(m^3)$. There are currently three different quantum algorithms proposed for the problem of GPR. 

Zhao et al. published \cite{zhaoQuantumassistedGaussianProcess2019} the original quantum GPR algorithm, employing HHL to perform the matrix inversion. The use of HHL also means that this algorithm relies on the assumption that the matrix $K$ is sparse in order to achieve an exponential speedup. 

Chen et al. \cite{chenQuantumAlgorithmGaussian2022} highlight some issues with Zhao's algorithm that they aim to rectify. They point out that Zhao's algorithm relies on the assumption that preparing $|k_*\rangle$ and simulating $K+\sigma_n^2$ can be done efficiently. Here $|k_*\rangle$ is a quantum state representing the classical vector $\mathbf{k}_*$. In fact, Chen's algorithm is able to remove the precomputation of $K$ completely, meaning that it takes only the training and test data as input rather than the covariance matrix. This is promising as it means that this algorithm makes very few assumptions about quantum state preparation. They achieve this by encoding the data as approximate coherent states. The advantage of this is that given two states encoding classical vectors $|x_i\rangle$ $|x_j\rangle$, we have $\langle x_i|x_j\rangle \approx k(x_i, x_j)$, where this approximation can be taken to arbitrary precision. However, one downside to this algorithm is that the kernel choice is limited to the RBF kernel and the number of qubits required seems to be high. Also, in the best case the runtime is $O(m\log(m))$, achieving only a small polynomial speedup.

The final algorithm by Farooq et al. \cite{farooqQuantumassistedHilbertspaceGaussian2024a} uses a Hilbert space basis function approximation to remove the dependence on sparsity that arises in Zhao's algorithm. This enables them to reduce the dependence on $m$ in the time complexity to logarithmic. We include the complexity of each algorithm in Table \ref{table1}.

\begin{table}[htbp]
\centering
\resizebox{0.5\textwidth}{!}{%
\renewcommand{\arraystretch}{2}
\begin{tabular}{|c|c|}
\hline
\textbf{Algorithm} & \textbf{Time Complexity} \\
\hline
Zhao et al.\cite{zhaoQuantumassistedGaussianProcess2019} & $O(\text{log}(m)\kappa^2s^2/\epsilon)$ \\
\hline
Chen et al.\cite{chenQuantumAlgorithmGaussian2022} & $O(\frac{1}{\sqrt{P_k}}d\text{log}(\frac{d}{\delta})m\text{log}(m)\epsilon^{-3}\kappa)$\\
\hline
Farooq et al.\cite{farooqQuantumassistedHilbertspaceGaussian2024a} & $O(\text{poly}(\text{log}(mM))\text{log}(M)\epsilon^{-3}\kappa^2)$ \\
\hline
\end{tabular}%
}
\caption{Comparison of complexities for quantum GPR algorithms. $\kappa$ and $s$ denote the condition number and sparsity of $K$ respectively. $M$ denotes the number of basis functions needed for the approximation in the algorithm of \cite{farooqQuantumassistedHilbertspaceGaussian2024a}. $P_k$ is the probability of correctly preparing the state $|k_*\rangle$ in the algorithm of \cite{chenQuantumAlgorithmGaussian2022}. Finally $\epsilon$ and $\delta$ both denote chosen precisions at various steps in each algorithm. }
\label{table1}
\end{table}
We note that all three of these algorithms have polynomial dependence on the condition number $\kappa$, and as such any claims about these algorithms providing a speedup are usually met with the assumption that the matrix is well conditioned, essentially meaning that the condition number scales as $O(\text{log}(m))$. In this paper we investigate how reasonable of an assumption this is and show that in a wide setting the condition number instead scales linearly with $m$.

We note that the dependency of quantum linear system solvers (QLSS) on the condition number has been improved after some of these algorithms for GPR have been published. The state-of-the-art is \cite{costa2021optimalscalingquantumlinear}, which has a runtime linear in $\kappa$, and a quantum algorithm for GPR that employs this paper's solver is similarly affected by our results.

The complexity given for Zhao et al. \cite{zhaoQuantumassistedGaussianProcess2019} in table \ref{table1} is for their original algorithm using HHL. The authors discuss how their algorithm can be generalised for use with QLSS other than HHL. We consider this in more detail in section \ref{discussion}.

\subsection{Previous Results on Condition Number Asymptotics}

Both the condition number of random matrices and the eigenvalues of kernel matrices have been studied somewhat extensively in the machine learning and pure mathematics literature. Edelman \cite{edelman_eigenvalues_1988} proves that the expected value of the log condition number of a matrix with normally distributed entries scales logarithmically with the matrix size. Tao and Vu \cite{tao_condition_2007} consider randomly perturbed matrices and are able to show that the condition number grows polynomially with the matrix size under some mild assumptions. These results are both indicative of an issue regarding the growth of the condition number, but do not consider kernel matrices specifically. Posa \cite{posa_conditioning_1989} performs a numerical analysis of the condition numbers of covariance matrices for several different kernels and Zimmerman \cite{zimmermann_condition_2015} investigates the effect of varying hyperparameters on the condition number, focusing on the RBF kernel and giving arguments as to why it often leads to ill-conditioned kernel matrices.

\section{Results}
\subsection{Main Theorems}
Let $(\chi,\mu)$ be a probability space and $k \in L_2(\chi^2,\mu^2)$ be a symmetric, measurable, kernel function. We define the corresponding integral operator $T_k : L_2(\chi,\mu) \to L_2(\chi, \mu)$ as follows:
$$T_k(f)(x') = \int_\chi k(x,x')f(x) \mathrm{d}\mu(x)$$

Then since $k$ is symmetric, $T_k$ is a self adjoint compact operator, so we can apply the following spectral decomposition. That is, the eigenvalues of $T_k$, denoted $\{ \lambda_i\}_{i=1}^\infty$, are real, and when written in order of descending absolute value, including multiplicities we have:
$$k(x,y) = \sum_{i=1}^\infty \lambda_i \phi_i(x)\phi_i(y)$$
where $\phi_i$ are the corresponding (normalised) eigenfunctions to each $\lambda_i$ and the above series converges in $L_2$ norm. Moreover we note that the eigenvalues $\{\lambda_i\}_{i=1}^\infty$ are square summable. These results are well known in spectral analysis, see for example \cite{chaitin1983spectral,dunford1988linear}.

Let $X_i:\Omega\to\chi$ be independent, $\mu$ distributed random variables. These are the inputs of the GPR algorithm. The matrix $K_m$ is given by $(K)_{i,j} = k(X_i,X_j)$ for $1 \leq i,j, \leq m$. Let $\tilde{\lambda}_{n,m}$ denote the $n$'th eigenvalue of the $m \times m$ matrix $K_m$, where the eigenvalues of $K_m$ are written in descending order. That is, for all $m\in\mathbb{N}$,
$$\tilde{\lambda}_{1,m}\geq\tilde{\lambda}_{2,m}\geq...\geq\tilde{\lambda}_{m,m}$$
Note that here the matrix $K_m$ is random and so the eigenvalues $\tilde{\lambda}_{n,m}$ are also technically random variables, although we suppress this notation for the most part and only discuss it where necessary. As above we denote the eigenvalues of the integral operator $T_k$ by $\lambda_n$ for $n \in \mathbb{N}$, where each eigenvalue occurs a number of times equal to its multiplicity and again they are taken to be in descending order. In our current notation, the condition number $\kappa_{m,\sigma}$ of the matrix $K_m + \sigma^2I$ is given by:
$$\kappa_{m,\sigma} = \frac{\tilde{\lambda}_{1,m} + \sigma^2}{\tilde{\lambda}_{m,m}+\sigma^2}$$
Recall that here $\sigma^2$ is often just a regularisation added to the diagonal in order to ensure that the kernel matrix is invertible, but that it also corresponds to the variance in the noise (as per the GPR description above).

We restrict ourselves to the scenario where as $m$ increases, the sampled data used to construct the matrix $K_m$ is kept. That is, instead of drawing new samples for the entire training data every time $m$ is increased, we keep the old data and instead sample one more data point to add. This assumption accurately reflects the way that the machine learning algorithm is implemented in practice \cite{rasmussenGaussianProcessesMachine2005}. We state this formally below and it will be assumed in the proof of theorems \ref{Thm:1}\&\ref{Thm:2}.
\begin{dataassumption}[Data Assumption]\label{ass:data}
    We assume the Gram matrices $K_m$ are constructed such that $\forall n<m$, $(K_n)_{i,j = 1}^n = k(x_i,x_j)$ and $(K_m)_{i,j = 1}^m = k(x_i,x_j)$ where $(x_i|i=1,...,n)$ contains the first $n$ elements of $(x_i|i=1,...,m)$ in the same order.
\end{dataassumption}
We present one theorem dealing with the $\sigma >0$ case and one for the $\sigma =0$ case.
\begin{theorem}\label{Thm:1}
    Let $k\in L_2(\chi^2,\mu^2)$ be a measurable, symmetric, positive semi-definite kernel function with bounded diagonal. Assume also that $\lambda_1\neq0$. Let $K_m$ be the corresponding Gram matrices constructed from $\mu$-distributed data according to the above data assumption. Finally, let $C$ denote the random variable obtained in the limit of $\tilde{\lambda}_{m,m}$ as $m\to\infty$. Then $\forall \sigma>0, \omega\in\Omega$,
    $$\kappa_{m,\sigma}(\omega) \sim \frac{\lambda_1}{C(\omega) + \sigma^2}m$$
    where the above holds with probability 1, in the sense that:
    $$\lim_{m\to\infty} \frac{C+\sigma^2}{\lambda_1m}\kappa_{m,\sigma} = 1 \hspace{1em}(a.s)$$
\end{theorem}
It is worth noting that since $C$ is a random variable, the coefficient in the asymptotic formula is not fully deterministic. That is, different samples $X_i$ may change the constant in the scaling behaviour (but it will still grow linearly with $m$). Moreover, $C$ is bounded in the sense that $\exists B>0$ s.t $\forall\omega\in\Omega, 0\leq C(\omega)\leq B$. This gives fully deterministic lower and upper bounds on the condition number. Explicitly:
\begin{corollary}\label{cor:1}
    Under the assumptions of Theorem 1, we have:
    $$\forall\sigma>0,\kappa_{m,\sigma}\in\Theta(m)$$
    or to be more precise:
    \begin{align}
            &\forall\sigma>0,\exists C_1,C_2,M>0 \hspace{0.5em}s.t\hspace{0.5em}\forall m\geq M,\forall\omega \in \Omega,\\
            &C_1m\leq\kappa_{m,\sigma}(\omega)\leq C_2m
    \end{align}
\end{corollary}
\textit{This conclusively shows that for a wide variety of kernels and domains, the condition number of $K+\sigma^2I$ grows linearly with the size of the matrix when $\sigma>0$. Moreover this result does not depend on the distribution of the input data, or the particular sample obtained.}

We note that the constraint of $\sigma>0$ is not too restricting here as this is an incredibly common practice within the machine learning community (since often the matrix $K$ will be singular).

One may be tempted to claim that the limit $C(\omega)$ in Theorem 1 is always 0. Indeed, numerics for a wide variety of kernels seem to indicate that this is true but we construct a counterexample in the supplementary material. 

Now we consider the case $\sigma=0$. In this scenario we impose that the kernel $k$ be strictly positive definite (rather than semi definite as above) We do this to ensure that the matrix $K$ is invertible, as if it is not the condition number instead has to be defined using the Moore-Penrose Pseudoinverse \cite{Penrose_1955} of $K$. Under this assumption the condition number of $K$ is given by:
$$\kappa_m = \frac{\tilde{\lambda}_{1,m}}{\tilde{\lambda}_{m,m}}$$
We have the following result:
\begin{theorem}\label{Thm:2}
    Let $k\in L_2(\chi^2,\mu^2)$ be a measurable, symmetric, positive definite kernel function with bounded diagonal. Let $K_m$ denote the $m\times m$ Gram matrix formed from data as in Lemma 2. Then with probability 1, the condition number of $K_m$ scales at least linearly with $m$.

    To be precise, the event $E\subset\Omega$ given by:
    $$E = \{\omega\in\Omega|\kappa_m(\omega)\in\Omega(m)\}
    $$
    occurs with probability 1. 
\end{theorem}
The lower bound in this case depends on the random variable $\frac{1}{\tilde{\lambda}_{1,1}(\omega)}$, see \nameref{Pf:thm2}. However, this is still sufficient to conclude that with probability 1, the condition number will scale at least linearly with $m$ for any sample.

Whilst positive definiteness of the kernel is enough to ensure $\tilde{\lambda}_{1,1}(\omega)>0, \forall\omega\in\Omega$ it is not enough to guarantee it is bounded away from 0. This means that we cannot achieve a deterministic lower bound (as in the $\sigma>0$ case) unless we place further regularity conditions on the kernel.

Together, these results are enough to show that the currently proposed quantum algorithms for GPR scale quadratically in the number of datapoints in a wide number of scenarios. In particular all three algorithms in table \ref{table1} provide no exponential speedup when compared to a classical approach. They also allow us to show the following corollary:
\begin{corollary}\label{cor:2}
    Let $k\in L_2(\chi^2,\mu^2)$ be a measurable symmetric positive (semi)-definite kernel function with bounded diagonal. Assume the matrices $K_m$ are constructed according to the \nameref{ass:data} and fix $\sigma \geq 0$. Then let $\kappa_F$ denote the condition number of $K_m + \sigma^2I$ with respect to the Frobenius norm. That is,
    $$\kappa_F := \norm{K_m+\sigma^2I}_F\norm{(K_m+\sigma^2I)^{-1}}_2$$
    where $\norm{\cdot}_F$ and $\norm{\cdot}_2$ denote the Frobenius and spectral norms respectively. Then:
    $$\kappa_F \in \Omega(m)$$
\end{corollary}

The collection of kernels, distributions and domains that these results apply to is broad. For example, any bounded kernel will automatically satisfy both $L_2(\mu)$ integrability and boundedness of the diagonal. Moreover, any continuous kernel on a compact domain will also satisfy these conditions by the extreme value theorem. This means that the result applies to a large number of the kernel functions used in machine learning applications. For example, the RBF kernel, the Matern kernel and the rational quadratic kernel \cite{rasmussenGaussianProcessesMachine2005} fulfill the conditions of this result on an arbitrary domain with arbitrary data distribution. However, unbounded continuous kernels, such as dot product and polynomial kernels, will not satisfy these conditions on non-compact domains (like $\mathbb{R}^d$).

We note that kernels on finite domains are automatically bounded and so fulfill the requirements of all our theorems when i.i.d. data is assumed. This situation arises in the context of graph kernels \cite{kriegeSurveyGraphKernels2020} and string kernels \cite{qiStringKernelsConstruction2022}.

In addition to our results on the condition number, closely related arguments can be used to show the following about the sparsity and Frobenius norm of kernel matrices.
\begin{theorem}\label{thm:3}
    Let $k\in L_2(\chi^2,\mu^2)$ be a measurable, bounded, symmetric, positive semi-definite kernel function with $\lambda_1\neq0$. Further let $s(m)$ denote the sparsity of the gram matrix $K_m$, that is the maximum number of non-zero entries in a row. Then with probability 1 we have:
    $$s(m) \in \Theta(m)$$
\end{theorem}
\begin{theorem}\label{thm:4}
    Let $k\in L_2(\chi^2,\mu^2)$ be a measurable, symmetric, positive semi-definite kernel function with bounded diagonal. 
    Moreover suppose $\lambda_1 \neq 0$ and let $\norm{\cdot}_F$ denote the Frobenius norm. Then $\forall\sigma\geq0$, we have the following with probability 1.
    $$\norm{K_m+\sigma^2}_F \in \Omega(m)$$
\end{theorem}
This is relevant to our discussion on GPR as Zhao's algorithm \cite{zhaoQuantumassistedGaussianProcess2019} depends quadratically on the sparsity of $K+\sigma_n^2I$. Also note that theorems \ref{thm:3}\&\ref{thm:4} do not rely on the \nameref{ass:data}.

In each of the theorems above we have assumed $\lambda_1\neq 0$ (indirectly in the case of theorem \ref{Thm:2} as this is implied by positive definiteness). This is a reasonable assumption as if $\lambda_1=0$ then it follows that all eigenvalues of $T_k$ are $0$ and hence the kernel is equal to $0$ $\mu^2$ almost everywhere on $\chi^2$. Therefore, this case is of little practical interest.
\subsection{Numerical Results}\label{section 2.2}
We provide numerical justification of the asymptotic formula for the condition number $\kappa_{m,\sigma}$ from theorem \ref{Thm:1}. In order to do this we require a kernel and measure for which an explicit formula of the first eigenvalue is known. A simple candidate for this is the RBF kernel with Gaussian measure, which is given by:
$$
k(x,y) = \exp\left(-\frac{1}{2l^2}\norm{x-y}^2\right)$$
where $l$ is a parameter representing the characteristic length scale of the kernel and \\
$x,y \in \chi \subset \mathbb{R}^d$. 

In the case of a $1d$ normal distribution there is an exact spectral decomposition of $T_k$ from Zhu et al. \cite{zhuGaussianRegressionOptimal1997}. 

We are able to generalise their results to a spectral decomposition of $T_k$ in the case of a non-degenerate, multivariate normal distribution. Recall, this corresponds to drawing the sample data $\{X_i\}_{i=1}^m$ for $(K)_{i,j} = k(X_i,X_j)$ from a distribution $N(\vec{v},\Sigma)$. For the proof of this generalisation see the supplementary material \ref{supp:2}. The result states the following; let $\{\nu_i\}_{i=1}^d$ denote the eigenvalues of $\Sigma$, note they are all positive as $\Sigma$ is positive definite. Then for the RBF kernel with the above normal density, the eigenvalues of the operator $T_k$ are indexed by a multi index $\vec{k} \in\mathbb{N}^d$ and are given by:
    $$\lambda_{\Vec{k}} = \prod_{i=1}^d\left(\frac{2a_i}{A_i}\right)^{1/2}B_i^{k_i}$$
where:
\begin{align}
    a_i &= \frac{1}{4\nu_i}\\
    b &= \frac{1}{2l^2}\\
    c_i &= \sqrt{a_i^2+2a_ib}\\
    A_i &= a_i + b + c_i\\
    B_i &= \frac{b}{A_i}
\end{align}
In particular we see that the largest eigenvalue occurs when $\vec{k}=\vec{0}$. It is a well known heuristic in the literature that the eigenvalues of $K_m$ decay to 0, as the size of the matrix ($m$) increases, for the RBF kernel with normally distributed input data. This can be easily verified numerically and motivates our choice of $C(\omega)\equiv0$ in the formula from theorem \ref{Thm:1}. This then gives the following predicted asymptotic behaviour for the condition number:
$$\kappa_{m,\sigma} \sim \frac{m}{\sigma_n^2}\prod_{i=1}^d\left(\frac{2a_i}{A_i}\right)^{1/2}$$
To perform our numerical analysis, for each $m$ we construct an $m\times m$ Gram matrix from normally distributed data and compute the condition number of $K+\sigma_n^2 I$ directly, we then plot this for different values of $m$. We plot on the same graph the predicted complexity formula above also. Note that both figures \ref{fig: 1} and \ref{fig: 2} are generated from a single realisation of the data. Multiple realisations show identical behaviour. 
\begin{figure}[ht]
    \centering
    \includegraphics[width=\textwidth]{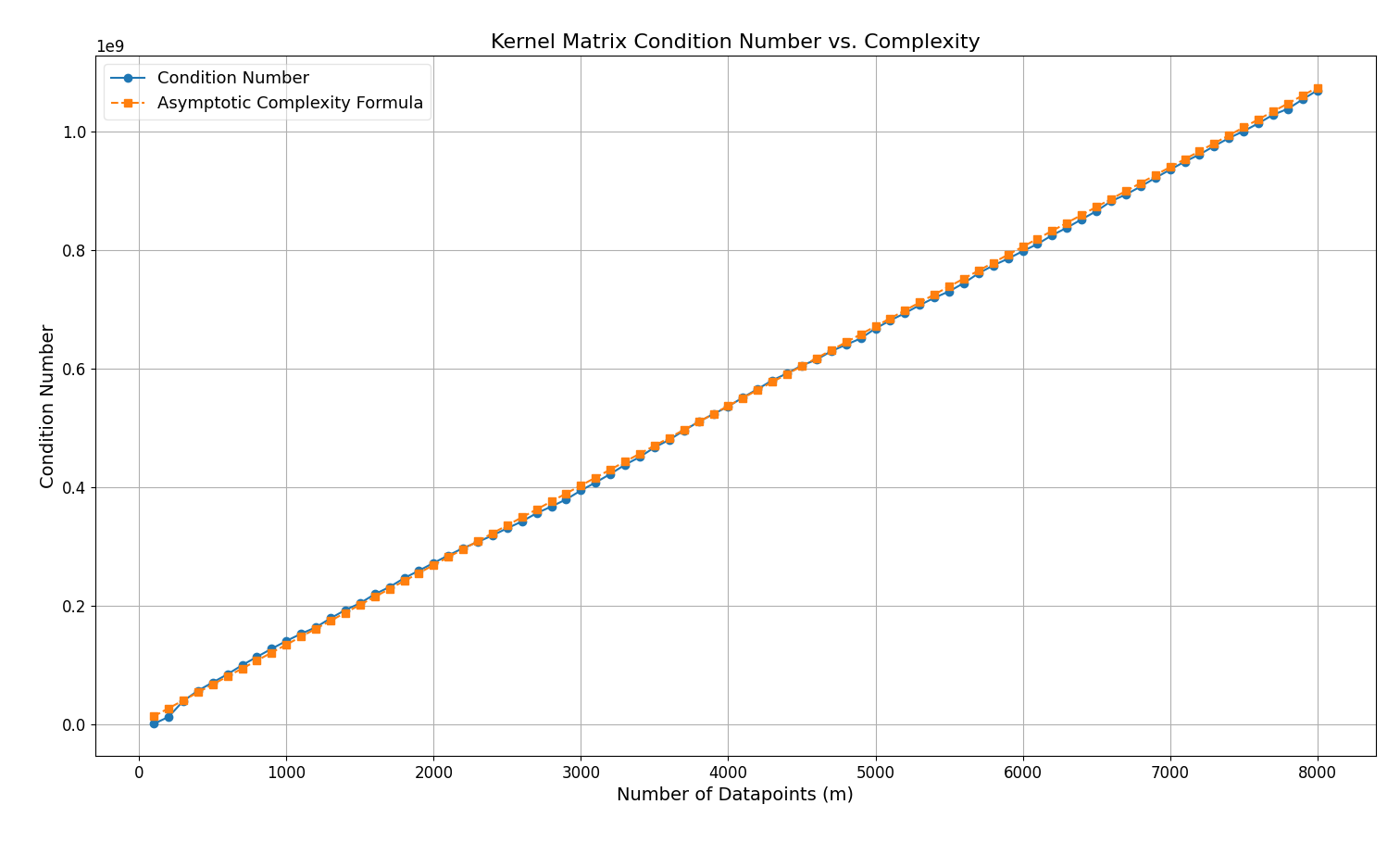}
    \caption{Condition numbers of Gram matrices formed from the RBF kernel with multivariate normal data vs predicted asymptotic behaviour. Parameter values: $\sigma_n=10^{-3}$, $l=1$.}
    \label{fig: 1}
\end{figure}

Figure \ref{fig: 1} shows the comparison between the predicted condition number and the actual condition number for $3d$ normally distributed data with mean $\vec{v}=\vec0$ and covariance matrix $\Sigma$, where 
\[
\Sigma = 
\begin{pmatrix}
1 & 0 & 0 \\
0 & 2.5 & -0.5 \\
0 & -0.5 & 2.5
\end{pmatrix}
\]

We see from figure \ref{fig: 1} that the asymptotic formula provides a very close approximation even for small values of $m$. 

As our results do not apply to unbounded kernels, we also perform numerics for the dot product kernel, $k(x,x') = x \cdot x'$, with the same input distribution as above. We plot the condition numbers as before and then perform a linear regression.

\begin{figure}[ht]
    \centering
    \includegraphics[width=\textwidth]{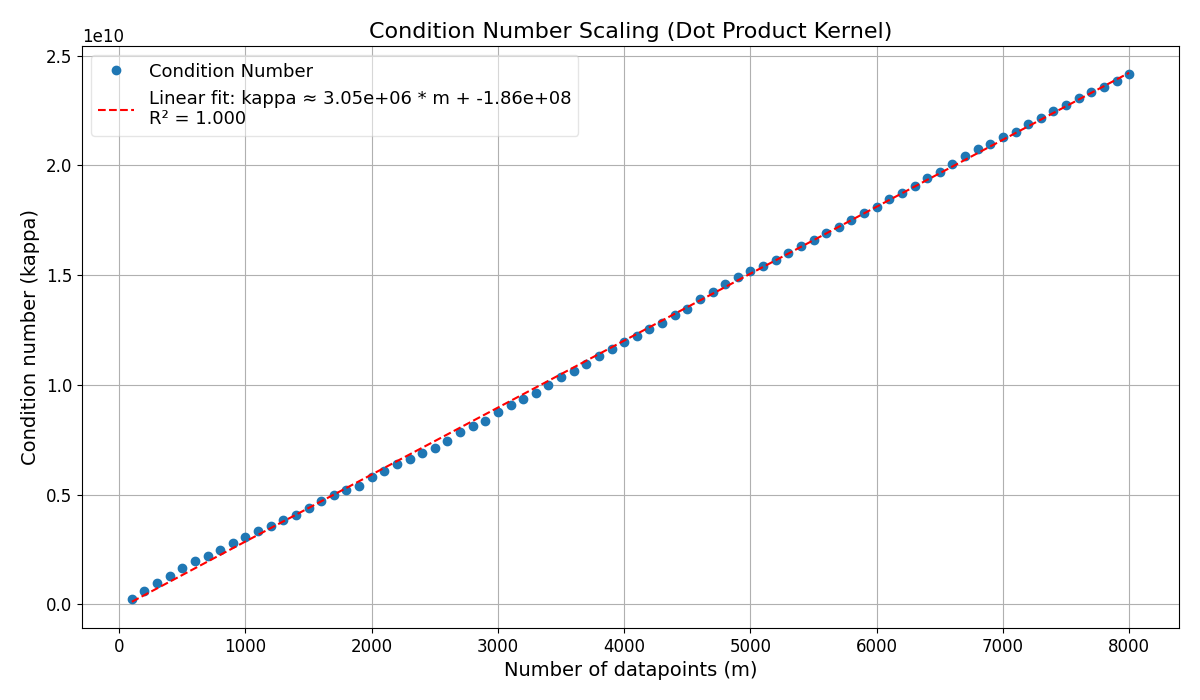}
    \caption{Condition number growth of Gram matrices formed using the dot product kernel with multivariate normal data. Parameter values: $\sigma_n=10^{-3}$.}
    \label{fig: 2}
\end{figure}

$R^2$ denotes the coefficient of determination and thus the value of 1 indicates a very strong linear correlation. It is clear that in this simulation the condition number scales linearly with the number of datapoints, despite the kernel being unbounded. This indicates that the linear scaling of $\kappa$ may extend beyond the assumptions in this paper. We therefore think it would be interesting to see whether similar results can be proved for unbounded kernels.

Note that both figures \ref{fig: 1} and \ref{fig: 2} were generated by drawing data in a way that enforces the \nameref{ass:data}. However we also performed numerics for the RBF kernel where new data was generated at each timestep. This plot can be seen in supplementary material \ref{supp:3} and it is clear that the linear scaling still holds. This provides evidence that similar results may be true without the \nameref{ass:data}.

\section{Discussion}\label{discussion}
We have proven that under a mild data assumption, a square integrable kernel with bounded diagonal and i.i.d. data will give a regularised Gram matrix whose condition number scales linearly in $m$ and unregularised Gram matrix whose condition number scales at least linearly in $m$. Moreover, if the kernel is additionally bounded then the sparsity also scales linearly in $m$. Under these assumptions, we can put the scaling of all three quantum GPR algorithms in terms of $m$ and obtain the complexities in Table \ref{table2}.
\begin{table}[htbp]
\centering
\resizebox{0.5\textwidth}{!}{%
\renewcommand{\arraystretch}{2}
\begin{tabular}{|c|c|}
\hline
\textbf{Algorithm} & \textbf{Updated Time Complexity} \\
\hline
Zhao\cite{zhaoQuantumassistedGaussianProcess2019} & $O(\text{log}(m)m^4/\epsilon)$ \\
\hline
Chen \cite{chenQuantumAlgorithmGaussian2022} & $O(\frac{1}{\sqrt{P_k}}d\text{log}(\frac{d}{\delta})m^2\text{log}(m)\epsilon^{-3})$\\
\hline
Farooq \cite{farooqQuantumassistedHilbertspaceGaussian2024a} & $O(\text{poly}(\text{log}mM)\text{log}(M)\epsilon^{-3}m^2)$ \\
\hline
Classical Cholesky & $O(m^3)$\\
\hline
\end{tabular}%
}
\caption{Comparison of updated complexities for quantum GPR algorithms, in the case of bounded kernels with i.i.d. data and $\sigma_n >0$.}
\label{table2}
\end{table}

We see that in this scenario these algorithms scale polynomially with the number of datapoints. Moreover, the best options are only linearly better than the standard classical approach. Therefore we can conclude that, for a wide variety of kernels, these algorithms offer at most a speedup from $O(m^3)$ to $O(m^2$log($m))$. We observe that this reduction in speedup occurs even if QRAM \cite{jaquesQRAMSurveyCritique2023} access is allowed for the quantum models, which is a reasonably strong assumption. This provides an argument against the construction of a QRAM in the case of quantum GPR.

The comparison in Table \ref{table2} is with exact matrix inversion but there are classical methods that give approximate solutions in similar time to the complexities listed. Reduced rank hilbert space methods \cite{solin_hilbert_2020} scale as $O(mM^2)$ where $M$ is the number of basis functions used for the approximation. Note, under the assumptions of theorems \ref{Thm:1} and \ref{thm:3} conjugate gradient methods \cite{ConjugateGradientDescent} scale as $O(m^3\log(1 / \epsilon))$. A review of further classical techniques can be found in \cite{liu_when_2019}.

The discussion for Zhao's algorithm above applies only to its original formulation using HHL. If a different QLSS is used then a better time complexity can be achieved. However, even the state of the art QLSS \cite{costa2021optimalscalingquantumlinear} has time complexity linear in $\kappa$ and so will also scale at least linearly with the number of datapoints when applied to GPR.

There exist QLSS algorithms that scale independently of the condition number and sparsity of a matrix \cite{wang_qubit-efficient_2024}. This algorithm relies on having a decomposition for the matrix in the Pauli basis. The runtime then depends quadratically on the $l_1$ norm of the coefficients in this decomposition. As a result, this algorithm avoids many of the issues we discuss in this paper and so would be an interesting candidate for a GPR algorithm.

We note that our results do leave the possibility for exponential advantage in the case of unbounded kernels, such as the dot product and polynomial kernels. It is not immediately clear how the condition number scales in this setting but numerics for the dot product kernel seem to indicate similar behaviour, see section \ref{section 2.2}.

There are similar possibilities in the case of non i.i.d data. Whilst i.i.d. data is a widely used assumption, it is clear that truncated kernels without i.i.d data can lead to sparsity and condition number that are polylogarithmic in $m$. These particular scenarios are therefore good candidates for quantum advantage.

There have been attempts to overcome poor condition number scaling in the literature, with a notable example being preconditioners \cite{claderPreconditionedQuantumLinear2013, shaoQuantumCirculantPreconditioner2018}. However, several quantum preconditioners are affected by our results also. For example, \cite{claderPreconditionedQuantumLinear2013} depends polynomially on the sparsity of $K$ and \cite{shaoQuantumCirculantPreconditioner2018} is polynomial in the Frobenius norm. Theorems \ref{thm:3} \& \ref{thm:4} therefore imply that a GPR algorithm relying on these will still scale polynomially in $m$. We note that there is broad literature on classical preconditioners \cite{wengerPreconditioningScalableGaussian2022} and so a hybrid approach incorporating classical preconditioners may lead to an overall speedup.

There are solutions to the condition number problem in some closely related areas. For example, Quantum Support Vector Machines \cite{rebentrost_quantum_2014} have a time complexity that depends on the `effective condition number'. Essentially, by only considering the matrix eigenvalues above some threshold it is possible to circumvent scaling with respect to the actual condition number. Our results show that this technique is unsuitable for GPR, as it will lead to an error scaling linearly in $m$. This is because the error is given by the square sum of the excluded eigenvalues, each of which is $O(m)$ by lemma \ref{lem:1}.

Due to the connection between GPR and other kernel methods, it is natural to consider how these results will impact the area of quantum kernels \cite{havlicek_supervised_2019,schuld_quantum_2019}. Provable advantage can be demonstrated by quantum kernels for artificial problems \cite{liu_rigorous_2021} but it is still unclear whether they exhibit quantum advantage in applicable scenarios \cite{huangPowerDataQuantum2021}. 

As our treatment of the domain $\chi$ has been completely general, our result does apply to Gram matrices formed from evaluating a quantum kernel. However, quantum kernel methods typically use a quantum device for the kernel evaluation and then perform the relevant linear algebra classically. In this case we hope that our results can help provide an educated choice as to which classical algorithm will be optimal to perform any required matrix inversion. In particular, when making this choice any condition number dependence should be treated as at least linear in the size of the matrix. 

We finally note that there is a connection between kernel methods and neural networks via the neural tangent kernel (NTK) \cite{jacotNeuralTangentKernel2020,leeWideNeuralNetworks2019}. Specifically, in the infinite width limit the NTK is deterministic and so our result will apply when it is bounded. In this scenario, the linear scaling of the condition number may affect the rate of convergence for training a neural network via gradient descent, as this depends on the smallest Gram matrix eigenvalue \cite{duGradientDescentProvably2019}.
\section{Methods}
\subsection{Statement and Proof of Lemmas}
In order to prove our results we require the use of two lemmas, which we state below:
\begin{lemma}\label{lem:1}
Let \( k \in L_2(\chi^2,\mu^2) \) be symmetric, measurable, positive semi-definite, and have bounded diagonal. Then,
\[
\lim_{m \to \infty} \left\| \left\{ \frac{1}{m} \tilde{\lambda}_{n,m} - \lambda_n \right\}_{n=1}^\infty \right\|_{l_2} = 0 \quad \text{(a.s.)}
\]
In particular, for any fixed \( n \in \mathbb{N} \),
\[
\frac{1}{m} \tilde{\lambda}_{n,m} \to \lambda_n \quad \text{(a.s.)} \quad \text{as } m \to \infty
\]
where convergence in both limits occurs almost surely.
\end{lemma}

\begin{lemma}\label{lem:2}
Under the \nameref{ass:data} regarding the construction of the matrices \( \{K_m\}_{m=1}^\infty \) and the assumptions of lemma \ref{lem:1}, the sequence \( \{\tilde{\lambda}_{m,m}\}_{m=1}^\infty \) is decreasing. Therefore, there exists a random variable \( C(\omega) \geq 0 \) such that for all $\omega\in\Omega$:
\[
\lim_{m \to \infty} \tilde{\lambda}_{m,m}(\omega) = C(\omega)
\]
\end{lemma}

Versions of lemma \ref{lem:1} have been acknowledged in the machine learning literature previously. This result appears in some form as theorem 3.4 in Baker \cite{bakerNumericalIntegrationTreatment1979}, but this treatment is fully deterministic and is proved using a Riemann sum argument. 

A treatment where the training inputs are treated as random variables $X_i$, meaning that the kernel matrix $K$ is random, is given by Koltchinskii and Giné \cite{koltchinskii_random_2000}. In this paper the authors predominantly focus on the adjusted matrix that is obtained by deleting the diagonal of the matrix $K$. They show that the normalised eigenvalues of this matrix converge to the eigenvalues of $T_k$ almost surely, if and only if the operator $T_k$ is Hilbert Schmidt. Note that $T_k$ is Hilbert Schmidt iff $k \in L_2(\chi^2,\mu^2)$.

Their proof of this result goes through for the matrix $K$ (without deleted diagonal) if we have the following:
$$\lim_{R\to \infty}\lim_{m\to \infty} \frac{1}{m^2}\sum_{i=1}^m(k-k_R)^2(X_i,X_i) = 0 \hspace{1em} (a.s)$$
Where the $X_i$ are independent, $\chi$ valued, $\mu$ distributed random variables, $k$ is a symmetric, real kernel satisfying the series expansion above and $k_R$ has the truncated expansion:
$$k_R(x,y) = \sum_{i=1}^R\lambda_i\phi_i(x)\phi_i(y)$$
It is easy to show that the above condition holds when the kernel $k$ has bounded diagonal, that is:
$$\exists B>0, \hspace{0.5em} \forall x\in \chi ,\hspace{0.5em} |k(x,x)|\leq B$$
From here the proof of lemma \ref{lem:1} is as follows:
\begin{proof}[Proof of Lemma~\ref{lem:1}]
    The convergence in $l_2$ norm follows immediately from the argument outlined above and Theorem 3.1 in \cite{koltchinskii_random_2000}, since the sequences $\{\lambda_n\}_{n=1}^\infty$ and $\{\tilde{\lambda}_{n,m}\}_{n=1}^\infty$ (where the latter is padded with 0s beyond $n=m+1)$ are both non-negative and decreasing. Non-negativity of both sequences of eigenvalues follows from $k$ being positive semi-definite.\\
    The final statement holds as, for any fixed $n \in \mathbb{N}$ and $\forall m\in \mathbb{N}$
    \begin{align}
        \left|\frac{1}{m}\tilde{\lambda}_{n,m} - \lambda_n\right|^2 &\leq \sum_{i=1}^\infty\left|\frac{1}{m}\tilde{\lambda}_{i,m} - \lambda_i\right|^2\\
        &=\norm{\{\frac{1}{m}\tilde{\lambda}_{n,m} - \lambda_n\}_{n=1}^\infty}_{l_2}\\
        &\to0 \hspace{1em} (a.s)
    \end{align}\\
\end{proof}
The proof of lemma \ref{lem:2} relies on the \nameref{ass:data}.
This assumption is helpful as it means that $\forall m\in \mathbb{N}$, $K_{m-1}$ is a principal submatrix of $K_m$. Recall that a principal submatrix is a submatrix obtained by repeatedly deleting rows and columns of the same index. This will allow us to apply the following theorem:
\begin{namedtheorem}[Cauchy Interlacing Theorem]
    Let $A$ be an $m\times m$ hermitian matrix with eigenvalues $\{\lambda_i\}_{i=1}^m$ and let $B$ be an $(m-1)\times(m-1)$ principal submatrix of $A$, with eigenvalues $\{\mu_i\}_{i=1}^{m-1}$. Assume both sets of eigenvalues are taken to be in decreasing orders and repeated according to their multiplicities. Then the two sets of eigenvalues are interlaced in the following way:
    $$\lambda_1 \geq \mu_1 \geq \lambda_2 \geq...\geq\lambda_{m-1}\geq\mu_{m-1}\geq\lambda_m$$
\end{namedtheorem}
This result is well known to the linear algebra literature and a short proof of the theorem is given in \cite{fisk2005shortproofcauchysinterlace}. This enables us to prove lemma \ref{lem:2}:
\begin{proof}[Proof of Lemma \ref{lem:2}]
    We show first that $\{\tilde{\lambda}_{m,m}\}_{m=1}^\infty$ is decreasing. 
    
    Note that $\forall m\geq 2$ the matrix $K_{m-1}$ is a principal submatrix of $K_m$. Moreover each matrix $K_m$ is real symmetric (and hence hermitian) so we can apply the Cauchy Interlacing Theorem. The last two terms of the inequality stated in the theorem give $\tilde{\lambda}_{m-1, m-1}\geq\tilde{\lambda}_{m,m}$, so the sequence is decreasing.
    
    Existence of the limit follows simply because $k$ being positive semi-definite implies that the matrices $K_m$ are also positive semi-definite $\forall m$, therefore we have $\tilde{\lambda}_{m,m} \geq 0$ $\forall m\in \mathbb{N}$. Hence the sequence $\{\tilde{\lambda}_{m,m}\}_{m=1}^\infty$ converges by the monotone convergence theorem. Non-negativity of $C$ follows by positive semi-definiteness of $K_m$ also.
\end{proof}
Recall that since each $\tilde{\lambda}_{m,m}$ is a random variable the limit $C$ is also a random variable $C:\Omega \to \mathbb{R}$. However a deterministic upper bound on $C(\omega)$ can be obtained as follows:
\begin{align}
    &\forall\omega\in\Omega,\forall m\in\mathbb{N}, \\
    &\begin{aligned}
            \tilde{\lambda}_{m,m}(w) &\leq \tilde{\lambda}_{1,1}(w)\\
    &= k(X_1(w),X_1(w))\\
    &\leq B
    \end{aligned}\\
    \\
    &\implies \forall \omega\in\Omega, C(\omega)\leq B
\end{align}
Where $B>0$ is the assumed bound on the diagonal of the kernel as above. Hence the random variable obtained in the limit is bounded between $0$ and $B$.

We remark that the convergence of $\tilde{\lambda}_{m,m}(\omega)$ to $C(\omega)$ holds $\forall\omega\in\Omega$, meaning it holds pointwise (or surely).
\subsection{Proof of Theorems}
Now using the lemmas we are ready to prove Theorems \ref{Thm:1} \& \ref{Thm:2}.
\begin{proof}[Proof of Theorem \ref{Thm:1}]
    The result immediately follows from Lemma 1 and 2, as if we let $A\subset\Omega$ such that:
    $$A = \left\{\omega\in\Omega|\lim_{m\to\infty}\frac{1}{m}\tilde{\lambda}_{1,m}(\omega) = \lambda_1 \right\}$$
    Note that $\lambda_1$ is deterministic. Then by Lemma 1, the event $A$ occurs with probability 1. Then $\forall\omega\in A$ we have:
    \begin{align}
        \lim_{m\to\infty} \kappa_{m,\sigma}(\omega)\frac{C(\omega
        ) +\sigma^2}{\lambda_1m} &= \lim_{m\to\infty}\frac{\tilde{\lambda}_{1,m}(\omega) + \sigma^2}{\tilde{\lambda}_{m,m}(\omega)+\sigma^2} \frac{C(w) +\sigma^2}{\lambda_1m}\\
        \\
        &=\lim_{m\to\infty}\frac{\frac{1}{m}\tilde{\lambda}_{1,m}(\omega)+\frac{1}{m}\sigma^2}{\tilde{\lambda}_{m,m} +\sigma^2}\frac{C(\omega)+\sigma^2}{\lambda_1}\\
        \\
        &=\frac{\lambda_1}{C(\omega) + \sigma^2}\frac{C(\omega)+\sigma^2}{\lambda_1}, \text{By definition of }A\text{ and Lemma 2}\\
        \\
        &= 1
    \end{align}
    Since the event $A$ occurs with probability 1 and the limit holds $\forall\omega\in A$, we conclude that the limit holds almost surely.
\end{proof}
The deduction of corollary \ref{cor:1} follows from this as:
\begin{align}
    &\forall\omega\in\Omega,\hspace{1em}0\leq C(\omega)\leq B\\
    \\
    &\implies \frac{\lambda_1}{B+\sigma^2}m\leq\frac{\lambda_1}{C(\omega)+\sigma^2}m\leq\frac{\lambda_1}{\sigma^2}m \hspace{1em}\forall m\in\mathbb{N}
\end{align}
\makeatletter
In addition we can prove theorem \ref{Thm:2} as follows:
\begin{proof}[Proof of Theorem \ref{Thm:2}]
\phantomsection
\protected@edef\@currentlabelname{Proof of Theorem~\ref{Thm:2}}
\label{Pf:thm2}
    Let $A\subset\Omega$ denote the following event:
    $$A = \left\{\omega\in\Omega|\lim_{m\to\infty}\frac{1}{m}\tilde{\lambda}_{1,m}(\omega) = \lambda_1(\omega)\right\}$$
    Where again the event occurs with probability 1 by Lemma 1. We also note that for any $\omega \in A$ we have the following by Lemma 2:
    $$\frac{1}{m}\kappa_m(\omega) = \frac{\tilde{\lambda}_{1,m}(\omega)}{m\tilde{\lambda}_{m,m}(\omega)} \geq \frac{\tilde{\lambda}_{1,m}(\omega)}{m\tilde{\lambda}_{1,1}(\omega)}\to\frac{\lambda_1}{\tilde{\lambda}_{1,1}(\omega)}$$
    Note that $\tilde{\lambda}_{1,1} = k(X_1,X_1)$. It follows from the above that for any fixed $\omega\in A$
    \begin{align}
        &\forall\epsilon>0,\exists M\in\mathbb{N} \hspace{0.5em}s.t\hspace{0.5em}\forall m\geq M,\\
        \\
        &\left|\frac{\frac{1}{m}\tilde{\lambda}_{1,m}}{\tilde{\lambda}_{1,1}} - \frac{\lambda_1}{\tilde{\lambda}_{1,1}} \right|<\epsilon\\
        \\
        &\implies \frac{\lambda_1}{\tilde{\lambda}_{1,1}} - \epsilon < \frac{\frac{1}{m}\tilde{\lambda}_{1,m}}{\tilde{{\lambda}}_{1,1}}
    \end{align}
    Taking $\epsilon = \frac{\lambda_1}{2\tilde{\lambda}_{1,1}}$ we then obtain for any fixed $\omega\in A$, $\exists M\in\mathbb{N}$ s.t $\forall
    m\geq M,$
    \begin{align}
        \frac{\lambda_1}{2\tilde{\lambda}_{1,1}(\omega)} &< \frac{\frac{1}{m}\tilde{\lambda}_{1,m}(\omega)}{\tilde{{\lambda}}_{1,1}(\omega)}\\
        \\
        &\leq \frac{1}{m} \frac{\tilde{\lambda}_{1,m}(\omega)}{\tilde{\lambda}_{m,m}
        (\omega)}\\
        &=\frac{1}{m}\kappa_m(\omega)
    \end{align}
    Hence $\forall \omega\in A,\kappa_m(\omega)\in\Omega(m)$ and so $\omega\in E$, which implies $A\subset E$. Since $A$ occurs with probability 1 we conclude that $E$ does also which completes the proof.
\end{proof}
\begin{proof}[Proof of Corollary \ref{cor:2}]
    For any matrix $A$ we have:
    $$\norm{A}_2\leq\norm{A}_F$$
    where $\norm{\cdot}_2$ denotes the spectral norm. This gives:
    \begin{align}
           \kappa &= \norm{K+\sigma^2I}_2\norm{(K+\sigma^2I)^{-1}}_2\\
           &\leq\norm{K+\sigma^2I}_F\norm{(K+\sigma^2I)^{-1}}_2\\
           &=\kappa_F
    \end{align}
    The result then follows in the $\sigma>0$ case by theorem \ref{Thm:1} and in the $\sigma=0$ case by theorem \ref{Thm:2}.
\end{proof}
\makeatother
\begin{proof}[Proof of Theorem \ref{thm:3}]
    Let $A = K+\sigma_n^2I$ and let $A_{i,j}$ denote the $i,j$'th element of $A$. Consider the largest eigenvalue of K, $\tilde{\lambda}_{1,1}$. Then by the Gershgorin Circle theorem $\exists i\in\{1,...,m\}$ such that:
    \begin{align}
        |\tilde{\lambda}_{1,m} - k(x_i,x_i)| &\leq \sum_{j\neq i}|k(x_i,x_j)|\\
        &\leq s(m)B
    \end{align}
    Where $B$ is the assumed bound on the kernel. Then by the triangle inequality this implies:
    $$|\tilde{\lambda}_{1,m}|\leq(s(m)+1)B$$
    Then by lemma 1 we have:
    \begin{align}
        \lim_{m\to\infty} \frac{(s(m)+1)B}{m} &\geq \lim_{m\to\infty}\frac{\tilde{\lambda}_{1,m}}{m}\\
        &=\lambda_1 >0
    \end{align}
    Where the above holds with probability 1. It then immediately follows that:
    $$\lim_{m\to\infty}\frac{1}{m}s(m)>0 \hspace{1em} (a.s)$$
    Finally, we have $s(m) \leq m$ by its definition and so we conclude that with probability 1:
    $$s(m) = \Theta(m)$$
\end{proof}
\begin{proof}[Proof of Theorem \ref{thm:4}]
    To see the lower bound, note that since $K+\sigma^2I$ is symmetric, positive semi-definite, its Frobenius norm is given by:
    \begin{align}
        \norm{K+\sigma^2I}_F^2 &= \sum_{i=1}^m(\tilde{\lambda}_{i,m} + \sigma^2)^2\\
        &\geq \tilde{\lambda}_{1,m}^2
    \end{align}
    Therefore,
    \begin{align}
        \frac{1}{m}\norm{K+\sigma^2I}_F &\geq \frac{1}{m}\tilde{\lambda}_{1,m}\\
        &\to\lambda_1 >0
    \end{align}
    Where the last line holds with probability 1 by lemma \ref{lem:1}.
\end{proof}
\bibliographystyle{unsrt}
\bibliography{references}

\begin{thebibliography}{10}

\bibitem{Biamonte_2017}
Jacob Biamonte, Peter Wittek, Nicola Pancotti, Patrick Rebentrost, Nathan Wiebe, and Seth Lloyd.
\newblock Quantum machine learning.
\newblock {\em Nature}, 549(7671):195–202, September 2017.

\bibitem{Cerezo_2022}
M.~Cerezo, Guillaume Verdon, Hsin-Yuan Huang, Lukasz Cincio, and Patrick~J. Coles.
\newblock Challenges and opportunities in quantum machine learning.
\newblock {\em Nature Computational Science}, 2(9):567–576, September 2022.

\bibitem{Ciliberto_2018}
Carlo Ciliberto, Mark Herbster, Alessandro~Davide Ialongo, Massimiliano Pontil, Andrea Rocchetto, Simone Severini, and Leonard Wossnig.
\newblock Quantum machine learning: a classical perspective.
\newblock {\em Proceedings of the Royal Society A: Mathematical, Physical and Engineering Sciences}, 474(2209):20170551, January 2018.

\bibitem{lloydQuantumAlgorithmsSupervised2013}
Seth Lloyd, Masoud Mohseni, and Patrick Rebentrost.
\newblock Quantum algorithms for supervised and unsupervised machine learning, November 2013.
\newblock arXiv:1307.0411.

\bibitem{lloydQuantumPrincipalComponent2014}
Seth Lloyd, Masoud Mohseni, and Patrick Rebentrost.
\newblock Quantum principal component analysis.
\newblock {\em Nature Physics}, 10(9):631--633, September 2014.
\newblock Publisher: Nature Publishing Group.

\bibitem{kerenidisQuantumRecommendationSystems2017}
Iordanis Kerenidis and Anupam Prakash.
\newblock Quantum {Recommendation} {Systems}.
\newblock {\em LIPIcs, Volume 67, ITCS 2017}, 67:49:1--49:21, 2017.
\newblock Artwork Size: 21 pages, 517106 bytes ISBN: 9783959770293 Medium: application/pdf Publisher: Schloss Dagstuhl – Leibniz-Zentrum für Informatik.

\bibitem{harrowQuantumAlgorithmSolving2009}
Aram~W. Harrow, Avinatan Hassidim, and Seth Lloyd.
\newblock Quantum algorithm for solving linear systems of equations, September 2009.
\newblock arXiv:0811.3171.

\bibitem{jaques2023qramsurveycritique}
Samuel Jaques and Arthur~G. Rattew.
\newblock Qram: A survey and critique, 2023.

\bibitem{tangQuantuminspiredClassicalAlgorithm2019}
Ewin Tang.
\newblock A quantum-inspired classical algorithm for recommendation systems, May 2019.
\newblock arXiv:1807.04271.

\bibitem{tangQuantumPrincipalComponent2021}
Ewin Tang.
\newblock Quantum principal component analysis only achieves an exponential speedup because of its state preparation assumptions, August 2021.
\newblock arXiv:1811.00414.

\bibitem{rasmussenGaussianProcessesMachine2005}
Carl~Edward Rasmussen and Christopher K.~I. Williams.
\newblock {\em Gaussian {Processes} for {Machine} {Learning}}.
\newblock The MIT Press, November 2005.

\bibitem{zhaoQuantumassistedGaussianProcess2019}
Zhikuan Zhao, Jack~K. Fitzsimons, and Joseph~F. Fitzsimons.
\newblock Quantum-assisted {Gaussian} process regression.
\newblock {\em Physical Review A}, 99(5):052331, May 2019.
\newblock Publisher: American Physical Society.

\bibitem{chenQuantumAlgorithmGaussian2022}
Meng-Han Chen, Chao-Hua Yu, Jian-Liang Gao, Kai Yu, Song Lin, Gong-De Guo, and Jing Li.
\newblock Quantum algorithm for {Gaussian} process regression.
\newblock {\em Physical Review A}, 106(1):012406, July 2022.

\bibitem{farooqQuantumassistedHilbertspaceGaussian2024a}
Ahmad Farooq, Cristian~A. Galvis-Florez, and Simo Särkkä.
\newblock Quantum-assisted {Hilbert}-space {Gaussian} process regression.
\newblock {\em Physical Review A}, 109(5):052410, May 2024.

\bibitem{ConjugateGradientDescent}
Jonathan~Richard Shewchuk.
\newblock An introduction to the conjugate gradient method without the agonizing pain, 1994.

\bibitem{shaoFasterQuantuminspiredAlgorithms2023}
Changpeng Shao and Ashley Montanaro.
\newblock Faster quantum-inspired algorithms for solving linear systems, April 2023.
\newblock arXiv:2103.10309 version: 2.

\bibitem{costa2021optimalscalingquantumlinear}
Pedro C.~S. Costa, Dong An, Yuval~R. Sanders, Yuan Su, Ryan Babbush, and Dominic~W. Berry.
\newblock Optimal scaling quantum linear systems solver via discrete adiabatic theorem, 2021.

\bibitem{edelman_eigenvalues_1988}
Alan Edelman.
\newblock Eigenvalues and {Condition} {Numbers} of {Random} {Matrices}.
\newblock {\em SIAM Journal on Matrix Analysis and Applications}, 9(4):543--560, October 1988.
\newblock Publisher: Society for Industrial and Applied Mathematics.

\bibitem{tao_condition_2007}
Terence Tao and Van Vu.
\newblock The condition number of a randomly perturbed matrix, March 2007.
\newblock arXiv:math/0703307 version: 1.

\bibitem{posa_conditioning_1989}
D.~Posa.
\newblock Conditioning of the stationary kriging matrices for some well-known covariance models.
\newblock {\em Mathematical Geology}, 21(7):755--765, 1989.

\bibitem{zimmermann_condition_2015}
R.~Zimmermann.
\newblock On the condition number anomaly of {Gaussian} correlation matrices.
\newblock {\em Linear Algebra and its Applications}, 466:512--526, February 2015.

\bibitem{chaitin1983spectral}
F.~Chaitin-Chatelin.
\newblock {\em Spectral approximation of linear operators}.
\newblock Computer science and applied mathematics. Academic Press, 1983.
\newblock tex.lccn: lc82006744.

\bibitem{dunford1988linear}
N.~Dunford and J.T. Schwartz.
\newblock {\em Linear operators, part 2: {Spectral} theory, self adjoint operators in hilbert space}.
\newblock Wiley classics library. Wiley, 1988.
\newblock tex.lccn: 88116254.

\bibitem{Penrose_1955}
R.~Penrose.
\newblock A generalized inverse for matrices.
\newblock {\em Mathematical Proceedings of the Cambridge Philosophical Society}, 51(3):406–413, 1955.

\bibitem{kriegeSurveyGraphKernels2020}
Nils~M. Kriege, Fredrik~D. Johansson, and Christopher Morris.
\newblock A {Survey} on {Graph} {Kernels}.
\newblock {\em Applied Network Science}, 5(1):6, December 2020.
\newblock arXiv:1903.11835 [cs].

\bibitem{qiStringKernelsConstruction2022}
Ren Qi, Fei Guo, and Quan Zou.
\newblock String kernels construction and fusion: a survey with bioinformatics application.
\newblock {\em Frontiers of Computer Science}, 16(6):166904, January 2022.

\bibitem{zhuGaussianRegressionOptimal1997}
Huaiyu Zhu, Christopher K.~I. Williams, Richard Rohwer, Michal Morciniec, and Michal Hammel.
\newblock Gaussian {Regression} and {Optimal} {Finite} {Dimensional} {Linear} {Models}, 1997.

\bibitem{jaquesQRAMSurveyCritique2023}
Samuel Jaques and Arthur~G. Rattew.
\newblock {QRAM}: {A} {Survey} and {Critique}, May 2023.
\newblock arXiv:2305.10310.

\bibitem{solin_hilbert_2020}
Arno Solin and Simo Särkkä.
\newblock Hilbert space methods for reduced-rank {Gaussian} process regression.
\newblock {\em Statistics and Computing}, 30(2):419--446, March 2020.

\bibitem{liu_when_2019}
Haitao Liu, Yew-Soon Ong, Xiaobo Shen, and Jianfei Cai.
\newblock When {Gaussian} {Process} {Meets} {Big} {Data}: {A} {Review} of {Scalable} {GPs}, April 2019.
\newblock arXiv:1807.01065 [stat].

\bibitem{wang_qubit-efficient_2024}
Samson Wang, Sam McArdle, and Mario Berta.
\newblock Qubit-{Efficient} {Randomized} {Quantum} {Algorithms} for {Linear} {Algebra}.
\newblock {\em PRX Quantum}, 5(2):020324, April 2024.

\bibitem{claderPreconditionedQuantumLinear2013}
B.~D. Clader, B.~C. Jacobs, and C.~R. Sprouse.
\newblock Preconditioned {Quantum} {Linear} {System} {Algorithm}.
\newblock {\em Physical Review Letters}, 110(25):250504, June 2013.
\newblock Publisher: American Physical Society.

\bibitem{shaoQuantumCirculantPreconditioner2018}
Changpeng Shao and Hua Xiang.
\newblock Quantum circulant preconditioner for a linear system of equations.
\newblock {\em Physical Review A}, 98(6):062321, December 2018.
\newblock Publisher: American Physical Society.

\bibitem{wengerPreconditioningScalableGaussian2022}
Jonathan Wenger, Geoff Pleiss, Philipp Hennig, John~P. Cunningham, and Jacob~R. Gardner.
\newblock Preconditioning for {Scalable} {Gaussian} {Process} {Hyperparameter} {Optimization}, June 2022.
\newblock arXiv:2107.00243 [cs].

\bibitem{rebentrost_quantum_2014}
Patrick Rebentrost, Masoud Mohseni, and Seth Lloyd.
\newblock Quantum {Support} {Vector} {Machine} for {Big} {Data} {Classification}.
\newblock {\em Physical Review Letters}, 113(13):130503, September 2014.
\newblock Publisher: American Physical Society.

\bibitem{havlicek_supervised_2019}
Vojtěch Havlíček, Antonio~D. Córcoles, Kristan Temme, Aram~W. Harrow, Abhinav Kandala, Jerry~M. Chow, and Jay~M. Gambetta.
\newblock Supervised learning with quantum-enhanced feature spaces.
\newblock {\em Nature}, 567(7747):209--212, March 2019.
\newblock Publisher: Nature Publishing Group.

\bibitem{schuld_quantum_2019}
Maria Schuld and Nathan Killoran.
\newblock Quantum {Machine} {Learning} in {Feature} {Hilbert} {Spaces}.
\newblock {\em Physical Review Letters}, 122(4):040504, February 2019.
\newblock Publisher: American Physical Society.

\bibitem{liu_rigorous_2021}
Yunchao Liu, Srinivasan Arunachalam, and Kristan Temme.
\newblock A rigorous and robust quantum speed-up in supervised machine learning.
\newblock {\em Nature Physics}, 17(9):1013--1017, September 2021.
\newblock arXiv:2010.02174 [quant-ph].

\bibitem{huangPowerDataQuantum2021}
Hsin-Yuan Huang, Michael Broughton, Masoud Mohseni, Ryan Babbush, Sergio Boixo, Hartmut Neven, and Jarrod~R. McClean.
\newblock Power of data in quantum machine learning.
\newblock {\em Nature Communications}, 12(1):2631, May 2021.
\newblock Publisher: Nature Publishing Group.

\bibitem{jacotNeuralTangentKernel2020}
Arthur Jacot, Franck Gabriel, and Clément Hongler.
\newblock Neural {Tangent} {Kernel}: {Convergence} and {Generalization} in {Neural} {Networks}, February 2020.
\newblock arXiv:1806.07572 [cs].

\bibitem{leeWideNeuralNetworks2019}
Jaehoon Lee, Lechao Xiao, Samuel~S. Schoenholz, Yasaman Bahri, Roman Novak, Jascha Sohl-Dickstein, and Jeffrey Pennington.
\newblock Wide {Neural} {Networks} of {Any} {Depth} {Evolve} as {Linear} {Models} {Under} {Gradient} {Descent}, December 2019.
\newblock arXiv:1902.06720 [stat].

\bibitem{duGradientDescentProvably2019}
Simon~S. Du, Xiyu Zhai, Barnabas Poczos, and Aarti Singh.
\newblock Gradient {Descent} {Provably} {Optimizes} {Over}-parameterized {Neural} {Networks}, February 2019.
\newblock arXiv:1810.02054 [cs].

\bibitem{bakerNumericalIntegrationTreatment1979}
Christopher T.~H. Baker.
\newblock Numerical {Integration} in the {Treatment} of {Integral} {Equations}.
\newblock In G.~Hämmerlin, editor, {\em Numerische {Integration}: {Tagung} im {Mathematischen} {Forschungsinstitut} {Oberwolfach} vom 1. bis 7. {Oktober} 1978}, pages 44--53. Birkhäuser, Basel, 1979.

\bibitem{koltchinskii_random_2000}
Vladimir Koltchinskii and Evarist Giné.
\newblock Random matrix approximation of spectra of integral operators.
\newblock {\em Bernoulli}, 6(1):113--167, February 2000.
\newblock Publisher: Bernoulli Society for Mathematical Statistics and Probability.

\bibitem{fisk2005shortproofcauchysinterlace}
Steve Fisk.
\newblock A very short proof of cauchy's interlace theorem for eigenvalues of hermitian matrices, 2005.

\end{thebibliography}
\paragraph{Acknowledgments}
This work was supported by the UK EPSRC (EP/W032643/1, EP/Y004752/1, and EP/Z53318X/1) , KIST Open Research programme and the National Research Foundation of Korea(NRF) grant funded by the Korea government(MSIT) (No. RS-2024-00413957). 
We acknowledge helpful discussions with Joe Fitzsimons.
\paragraph{Author Contributions}
All three authors conceived the idea, DL performed the calculations, proved the theorems and wrote the first draft. All three participated in discussions and writing the final draft. The work was supervised by RB and MSK.
\paragraph{Data Availability}
The code used to generate Figures \ref{fig: 1},\ref{fig: 2} and \ref{fig: 3} is available at:\\
\href{https://github.com/Dominic-Lowe/GPR-Condition-Number}{https://github.com/Dominic-Lowe/GPR-Condition-Number}
\paragraph{Competing Interests}
The authors declare no competing interests.
\paragraph{Correspondence}
Any correspondence can be written to Dominic Lowe (dal20@ic.ac.uk), M.S. Kim (m.kim@imperial.ac.uk) and Roberto Bondesan (r.bondesan@imperial.ac.uk).

\newpage
\appendix
\renewcommand{\thesection}{S\arabic{section}} 
\section*{Supplementary Material}
\addcontentsline{toc}{section}{Supplementary Material}
\setcounter{section}{0}

\section{Counterexample to C = 0}\label{supp:counterexample}
\begin{equation}  
k(x,y) = \begin{cases} 
      1+e^{-x^2} & x=y \\
      0 & \text{otherwise}
   \end{cases}
\end{equation}
The above kernel is in $L_2(\chi^2,\mu^2)$ when $\mu$ is absolutely continuous with respect to the Lebesgue measure, since the diagonal will have measure 0 in $\chi^2$ and so $k =0$ almost everywhere. The kernel is then clearly positive semi-definite and has bounded diagonal. However, we note that the Gram matrices $K_m$ will be diagonal with entries strictly greater than 1, for any $m$. Thus we are able to conclude that the limit of the smallest eigenvalues, i.e $\tilde{\lambda}_{m,m}$, cannot be smaller than 1 for any distribution or sample. That is, $C(\omega)\geq1, \forall \omega\in\Omega$.

We were unable to obtain a counterexample that is not zero almost everywhere, so it is possible that the limit is provably zero with this restriction. We leave this as an open question.

\section{Spectral Decomposition of RBF Kernel with 
Multivariate Normal Measure}\label{supp:2}
We now show how the decomposition of the integral operator $T_k$ in the case of the RBF kernel with 1D normal density can be extended to a general $d$-dimensional normal distribution. Initially we will assume each vector $x \in \mathbb{R}^d$ is drawn from an isotropic Gaussian distribution $p(x) \sim N(\vec{v}, \sigma_x^2I)$ and we denote the $i$'th entry of $x$ by $x_i$. We will then generalise this to the multivariate case. The density $p(x)$ is given by:  
\begin{align}
    p(x) &= (2\pi)^{-d/2}\text{det}(\sigma_x^2I)^{-1/2}\text{exp}(-\frac{1}{2}(x-\vec{v})^\top(\sigma_x^2I)^{-1}(x-\vec{v}))\\
    \\
    &=\frac{1}{(2\pi)^{d/2}\sigma_x^d}\text{exp}\left(-\frac{(\norm{x- \vec{\mu}}^2}{2\sigma_x^{2n}}\right)\\
    \\
    &= p_1(x_1)...p_d(x_d) 
\end{align}
where we have taken:
$$p_i(x_i) = \frac{1}{\sqrt{2\pi}\sigma_x}\text{exp}\left(-\frac{x_i-v_i}{2\sigma_x^2}\right)$$
Note that $p_i$ is just the density of the distribution $N(v_i, \sigma_x^2)$, the exact setting described in the 1D case. Then in $d$ dimensions we see that the operator $T_k$ can be rewritten as:
\begin{align}
    &T_k(f)(y) = \int_{\mathbb{R}^D} \text{exp}(\frac{-1}{2l^2}\norm{x-y}^2)f(x) p(x) dx\\
    \\
    &= \int_{-\infty}^\infty...\int_{-\infty}^\infty e^{\frac{-1}{2l^2}\norm{x_1 - y_1}^2}...e^{\frac{-1}{2l^2}\norm{x_d - y_d}^2}p_1(x_1)...p_d(x_d)f(x) dx_1dx_2...dx_d
\end{align}
We see from the above that when $f(x)$ is separable the above reduces to a product of 1D integrals. Therefore it follows that products of the known 1D eigenfunctions give eigenfunctions of $T_k$ in $d$ dimensions. Specifically given a multi index $\vec{k} = (k_1,...,k_d)\in\mathbb{N}^d$ we have a corresponding eigenfunction:
$$\prod_{i=1}^d\phi_{k_i}(x_i-\mu_i)$$
The eigenvalue corresponding to this is then given by:
$$\lambda_{\vec{k}} = (\frac{2a}{A})^{d/2}B^{\tilde{k}}$$
Where $\tilde{k} = \sum_{i=1}^dk_i$. As each eigenvalue only depends on $\tilde{k}$ we index them using this instead of the multi-index $\Vec{k}$. It then becomes clear that each eigenvalue $\lambda_{\tilde{k}}$ has a degeneracy given by the number of ways a sum of $d$ non-negative integers can equal $\tilde{k}$. This degeneracy can be expressed as:
\[
\binom{\tilde{k} + d - 1}{d - 1}
\]
Using the above we can extend the argument to a multivariate normal distribution. That is, the case $p(x) \sim N(\vec{v}, \Sigma)$, where the only assumption we make is that $\text{det}(\Sigma)\neq0$. This is equivalent to assuming the Gaussian distribution is not degenerate and also that it has a valid density function with respect to the Lebesgue measure. This implies that $\Sigma$ is symmetric positive definite and therefore has strictly positive eigenvalues, which we denote by $\{\nu_i\}_{i=1}^{d}$. It is worth noting that these eigenvalues need not be distinct but the matrix is still diagonalisable. We therefore write:
$$\Sigma = PDP^{-1}$$
Where $D = \text{diag}(\{\nu_i\}_{i=1}^d)$ and $P$ is defined in the usual way. Note that since $\Sigma$ is symmetric $P$ is orthogonal, so $P^\top = P^{-1}$. The density $p(x)$ is then given by:
$$p(x) = \frac{1}{(2\pi)^{d/2}(\nu_1...\nu_d)^{1/2}}\text{exp}(\frac{-1}{2}(x - \vec{v})^\top P^\top D^{-1}P(x-\vec{v}))$$
The eigenvalue equation for the integral operator $T_k$ is then given by:
$$\frac{1}{(2\pi)^{d/2}(\nu_1...\nu_d)^{1/2}}\int_{\mathbb{R}^d}\text{exp}(\frac{-1}{2l^2}\norm{x-y}^2)\text{exp}(\frac{-1}{2}(x - \vec{v})^\top P^\top D^{-1}P(x- \vec{v}))\phi_k(x)\hspace{0.1cm}dx = \lambda_k\phi_k(y)$$
For all $y\in \mathbb{R}^d$.Here we perform the change of variables $z = P(x - \vec{v})$ and we can also replace $y$ with $P^{-1}y+\vec{v}$ as the equation should hold for all $y\in \mathbb{R}^d$. This allows us to rewrite the equation as:
$$\frac{1}{(2\pi)^{d/2}(\nu_1...\nu_d)^{1/2}}\int_{\mathbb{R}^d}\text{exp}(\frac{-1}{2l^2}\norm{P^{-1}z - P^{-1}y}^2)\text{exp}(\frac{-1}{2}z^\top D^{-1}z)\phi_k(P^{-1}z + \vec{v})\hspace{0.1cm}dz = \lambda_k\phi_k(P^{-1}y + \vec{v}) $$
Where the Jacobian of the $x\to z$ transformation is $P$ which has determinant 1 in absolute value, since it is orthogonal. Then since $P^{-1}$ is also orthogonal it follows that $\norm{P^{-1}(z-y)} = \norm{z-y}$. Hence with this we can simplify the left hand side of the equation as:
\begin{align}
    &\frac{1}{(2\pi)^{d/2}(\nu_1...\nu_d)^{1/2}}\int_{\mathbb{R}^d}\text{exp}(\frac{-1}{2l^2}\norm{P^{-1}z - P^{-1}y}^2)\text{exp}(\frac{-1}{2}z^\top D^{-1}z)\phi_k(P^{-1}z + \vec{v}) \hspace{0.1cm} dz \\
    \\
    &= \frac{1}{(2\pi)^{d/2}(\nu_1...\nu_d)^{1/2}}\int_{\mathbb{R}^d}\text{exp}(\frac{-1}{2l^2}\norm{z - y}^2)\text{exp}\left(\frac{-1}{2}\sum_{i=1}^d\frac{z_i^2}{\nu_i}\right)\phi_k(P^{-1}z + \vec{v}) \hspace{0.1cm}dz\\
    \\
    &=\int_{\mathbb{R}^d}\prod_{i=1}^d\left(\frac{1}{(2\pi\nu_i)^{1/2}}\text{exp}(\frac{-1}{2l^2}(z_i-y_i)^2)\text{exp}(\frac{-z_i^2}{2\nu_i})\right)\phi_k(P^{-1}z+\vec{v})\hspace{0.1cm}dz\\
    \\
    &=\int_{\mathbb{R}^d}\prod_{i=1}^d\left(\text{exp}(\frac{-1}{2l^2}(z_i-y_i)^2)p_i(z_i)\right)\phi_k(P^{-1}z+\vec{v})\hspace{0.1cm}dz\\
\end{align}
Where here each $p_i(z_i)$ is the density of a 1D gaussian $N(0,\nu_i)$. As before it is clear that if we take $\phi_k(P^{-1}z + \vec{v})$ to be a product of the 1d eigenfunctions for a $N(0,\nu_i)$ distribution we get an eigenfunction. In this case the eigenvalue will then also be the product of the corresponding 1d eigenvalues. It is clear that these higher dimensional eigenfunctions are better indexed by a multi index $\vec{k} \in \mathbb{N}^d$. If we also denote the 1d eigenfunctions as $\tilde{\phi}$ then they are given by:
    $$\tilde{\phi}_{k_i}(x) = \text{exp}(-(c_i-a_i)x^2)H_{k_i}(x\sqrt{2c_i})$$
Where $H_{k_i}$ denotes the $k_i$'th Physicist's Hermite polynomial. Then we see that the higher dimensional eigenfunctions $\phi_{\vec{k}}$ are given by:
$$\phi_{\vec{k}}(P^{-1}x+\vec{v}) = \prod_{i=1}^d\tilde{\phi}_{k_i}(x_i)$$
or equivalently:
$$\phi_{\vec{k}}(x) = \prod_{i=1}^d\tilde{\phi}_{k_i}([P(x-\vec{v})]_i)$$
The corresponding eigenvalue is then given by:
$$\lambda_{\vec{k}} = \prod_{i=1}^d\left(\frac{2a_i}{A_i}\right)^{1/2}B_i^{k_i}$$
where:
\begin{align}
    a_i &= \frac{1}{4\nu_i}\\
    b &= \frac{1}{2l^2}\\
    c_i &= \sqrt{a_i^2+2a_ib}\\
    A_i &= a_i + b + c_i\\
    B_i &= \frac{b}{A_i}
\end{align}
\section{Plot Without Data Assumption}\label{supp:3}
\begin{figure}[ht]
    \centering
    \includegraphics[width=\textwidth]{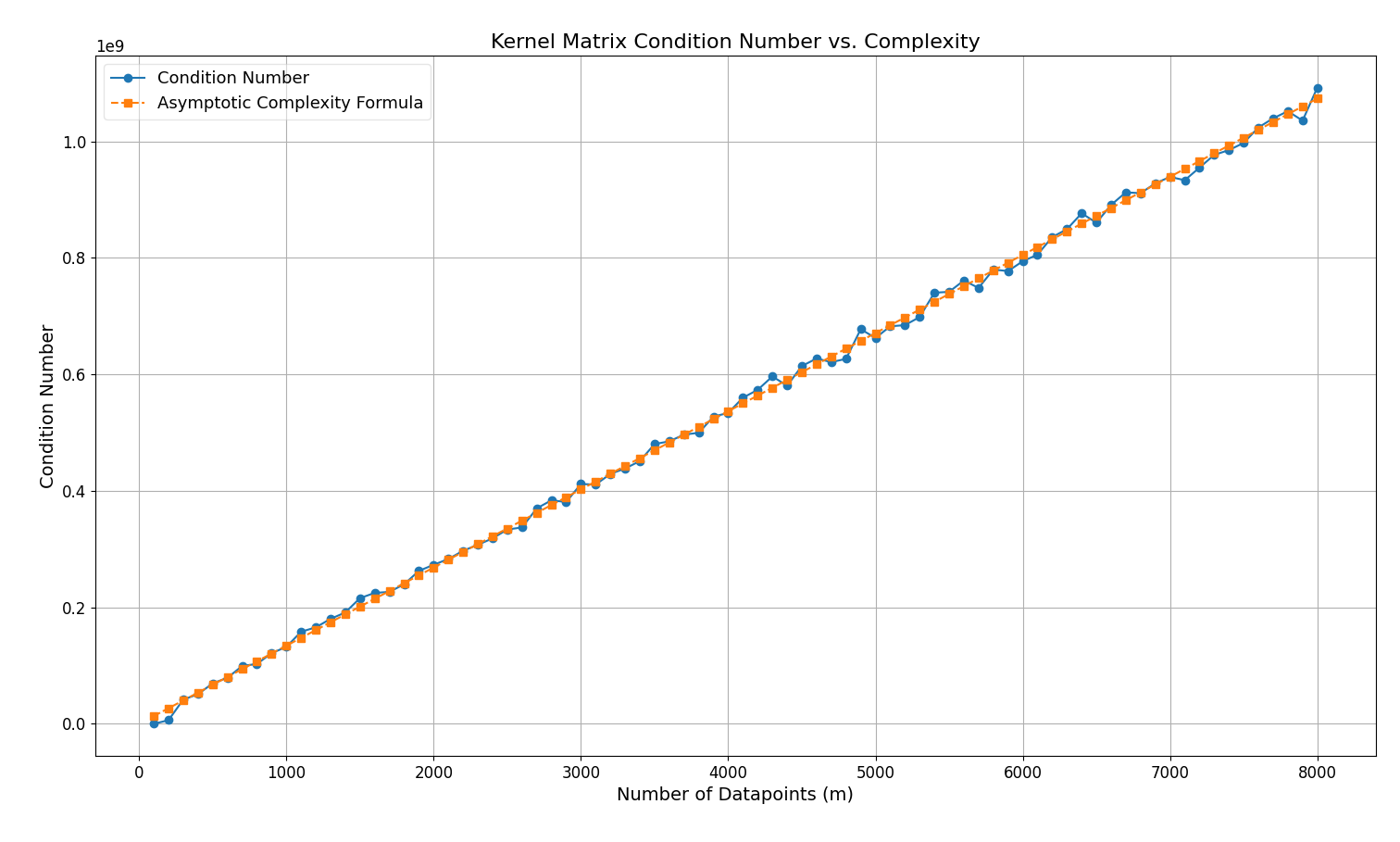}
    \caption{Condition number growth of Gram matrices formed using the RBF kernel with multivariate normal data $N(\vec{v},\Sigma)$, where entirely new data is generated at each timestep. That is, without the \nameref{ass:data}. Parameter values: $\vec{v}=\vec{0}$ $\sigma_n=10^{-3}$, $l=1$, and $\Sigma$ is as stated in section \ref{section 2.2}}
    \label{fig: 3}
\end{figure}

\end{document}